\documentclass[11pt,letter,reqno]{article}
\usepackage[utf8]{inputenc}

\title{A Lower Bound for Polynomial Calculus with Extension Rule}
\author{Yaroslav Alekseev\thanks{Steklov Institute of Mathematics at St.~Petersburg, St.~Petersburg, Russia, and Chebyshev Laboratory at St.~Petersburg State University}}
\date{}

\usepackage[square,sort,comma,numbers]{natbib}
\usepackage{xifthen}
\usepackage{graphicx}
\usepackage{amsthm}
\usepackage{amsmath}
\usepackage{amssymb}
\usepackage{ dsfont }
\usepackage{latexsym}
\usepackage{xcolor}
\usepackage{hyperref}
\usepackage[margin=1in]{geometry}
\usepackage{euscript}
\usepackage{a4wide}
\usepackage{enumitem}

\usepackage[notransparent]{svg}
\usepackage{boxedminipage}

\usepackage[T2A]{fontenc}
\usepackage[utf8]{inputenc}
\usepackage[english]{babel}

\definecolor{linkcolor}{HTML}{799B03}
\definecolor{urlcolor}{HTML}{799B03}

\hypersetup{pdfstartview=FitH,  linkcolor=linkcolor,urlcolor=urlcolor, colorlinks=true}

\newtheorem{definition}{Definition}

\newtheorem*{note}{Note}
\newtheorem*{lemma}{Lemma}

\newtheorem{theorem}{Theorem}[section]
\newtheorem{claim}[theorem]{Claim}
\newtheorem{corollary}[theorem]{Corollary}
\newtheorem*{claim*}{Claim}
\newtheorem*{claim-3-4}{Claim 3.4}

\newcommand {\para}[1] {\paragraph{#1}}
\newenvironment{proofclaim}{\par\noindent{\it Proof of claim}:}
{\vrule width 1ex height 1ex depth 0pt $_{\textrm{ Claim}}$ \vspace{0.5\baselineskip}}

\usepackage{nth}
\usepackage{intcalc}
\usepackage{etoolbox}
\usepackage{xstring}

\newcommand{\cSTOC}[1]{\nth{\intcalcSub{#1}{1968}}\ Annual\ ACM\ Symposium\ on\ Theory\ of\ Computing\ (STOC\ #1)}

\newcommand{\ComputationalComplexity}{Computational Complexity}

\newcommand{\STOC}[1]{Proceedings\ of\ the\ \cSTOC{#1}}

\newcommand{\ECCCurl}[2]{http://eccc.hpi-web.de/report/\ifnumcomp{#1}{>}{93}{19}{20}#1/#2/}

\newcommand{\shortECCC}[2]{\texttt{\href{http://eccc.hpi-web.de/report/\ifnumcomp{#1}{>}{93}{19}{20}#1/#2/}{eccc:TR#1-#2}}}

\newcommand{\parseECCC}[1]{
\StrSubstitute{#1}{TR}{}[\tmpstring]%
\IfSubStr{\tmpstring}{/}{ 
\StrBefore{\tmpstring}{/}[\ecccyear]%
\StrBehind{\tmpstring}{/}[\ecccreport]%
}{
\StrBefore{\tmpstring}{-}[\ecccyear]%
\StrBehind{\tmpstring}{-}[\ecccreport]%
}%
\shortECCC{\ecccyear}{\ecccreport}}

\begin{document}
\maketitle
\thispagestyle{empty}
\begin{abstract}
In this paper we study an extension of the Polynomial Calculus proof system where we can introduce new variables and take a square root. We prove that an instance of the subset-sum principle, the bit-value principle  $1 + x_1 + 2 x_2 + \ldots 2^{n - 1} x_n = 0$ ($\mathsf{BVP}_n$), requires refutations of exponential bit size over $\mathbb{Q}$ in this system. 

Part and Tzameret \cite{PT18_new} proved an exponential lower bound on the size of $\mathsf{Res}$-$\mathsf{Lin}$ (Resolution over linear equations \cite{RT07}) refutations of $\mathsf{BVP}_n$. We show that our system p-simulates $\mathsf{Res}$-$\mathsf{Lin}$ and thus we get an alternative exponential lower bound for the size of $\mathsf{Res}$-$\mathsf{Lin}$ refutations of $\mathsf{BVP}_n$.
\end{abstract}

\section{Introduction}
In essence, the study of propositional proof complexity started with the work of Cook and Reckhow \cite{CR79}.   The first superpolynomial bound on the proof size was proved in a pioneering work of Tseitin \cite{Tse68} for regular resolution.    Since then, many proof systems have been studied, some of them are logic-style (working with disjunctions, conjunctions and other Boolean operations) and some of them are algebraic (working with arbitrary polynomials).

In this work, we consider extensions of two systems, an algebraic one and a logic-style one. 
\para{Algebraic proof systems.}
Lower bounds for algebraic systems started with an exponential lower bound for the $\mathsf{Nullstellensatz}$ \cite{BeameIKPP96} system. The main system considered in this paper is based on the $\mathsf{Polynomial}$ $\mathsf{Calculus}$ system \cite{CEI96}, which is a dynamic version of $\mathsf{Nullstellensatz}$. Many exponential lower bounds are known for the size of $\mathsf{Polynomial}$ $\mathsf{Calculus}$ proofs for tautologies like the Pigeonhole Principle \cite{MR1691494,IPS99} and Tseitin tautologies \cite{BussGIP01}. While most results concern the representation of Boolean values by 0 and 1, there are also exponential lower bounds over the $\{-1, +1\}$ basis \cite{Sok20}.

Many extensions of $\mathsf{Polynomial}$ $\mathsf{Calculus}$ and $\mathsf{Nullstellensatz}$ have been considered before. Buss et al. \cite{BussIKPRS96} showed that there is a  tight connection between the lengths of constant-depth Frege proofs with $MOD_p$ gates and the length of $\mathsf{Nullstellensatz}$ refutations using extension axioms. Impagliazzo,  Mouli and Pitassi \cite{IMP19_new} showed that a depth-3 extension of $\mathsf{Polynomial}$ $\mathsf{Calculus}$ called $\Sigma\Pi\Sigma$-$\mathsf{PC}$ p-simulates $\mathsf{CP}^*$ (an inequalities-based system, $\mathsf{Cutting}$ $\mathsf{Planes}$ \cite{CCT87, CHVATAL1989455} with coefficients written in unary) over $\mathbb{Q}$. Also, they showed that a stronger extension of $\mathsf{Polynomial}$ $\mathsf{Calculus}$, called $\mathsf{Depth}$-$k$-$\mathsf{PC}$, p-simulates $\mathsf{Cutting}$ $\mathsf{Planes}$ and another inequalities-based system $\mathsf{Sum}$-$\mathsf{of}$-$\mathsf{Squares}$; the simulations can be conducted over $\mathbb{F}_{p^m}$ for arbitrary prime number $p$ if $m$ is sufficiently large.

Also very strong extensions were considered: Grigoriev and Hirsch \cite{GH03} considered algebraic systems over formulas. Grochow and Pitassi \cite{GP14} introduced the Ideal Proof System, $\mathsf{IPS}$, which can be considered as the version of $\mathsf{Nullstellensatz}$ where all polynomials are written as  algebraic circuits (see also \cite{Pit97, Pit98} for earlier versions of this system).

\para{Logic-style systems.}
While exponential lower bounds for low-depth proof systems (both algebraic and logical ones) are known for decades, the situation with higher depth proof systems is much worse. The present knowledge is limited to exponential bounds for constant-depth Frege systems over de Morgan basis (that is, without xor's or equivalences) \cite{Ajt94, BussIKPRS96, BeameIKPP96}. In particular, no truly exponential lower bounds are known for the size of refutations of formulas in CNF in (dag-like) systems that work over disjunctions of equations or inequalities (see \cite{Kra98-Discretely} as the first paper defining these systems and containing partial results). $\mathsf{Res}$-$\mathsf{Lin}$ (defined in \cite{RT07}), working with disjunctions of linear equations, is the second system considered in our paper, and it can be viewed as a generalization of Resolution. Part and Tzameret \cite{PT18_new} proved an exponential lower bound for (dag-like) $\mathsf{Res}$-$\mathsf{Lin}$ refutations over $\mathbb{Q}$ for the bit-value principle $\mathsf{BVP}_n$. Although this is the first exponential lower bound for this system, the instance does not constitute a translation of a formula in CNF.  Itsykson and Sokolov \cite{ITSYKSON2020102722} consider another extension of the resolution proof system that operates with disjunctions of linear equalities over $\mathbb{F}_2$ named $\mathsf{Res}(\oplus)$ and proved an exponential lower bound on the size of tree-like $\mathsf{Res}(\oplus)$-proofs. 

\subsection{Our results}
We extend $\mathsf{Polynomial}$ $\mathsf{Calculus}$ with two additional rules. One rule allows to take a square root (it was introduced by Grigoriev and Hirsch \cite{GH03} in the context of transforming refutation proofs of non-Boolean formulas into derivation proofs; our motivation to take square roots is to consider an algebraic system that is at least as strong as $\mathsf{Res}$-$\mathsf{Lin}$ even for non-Boolean formulas, see below). Another rule is an algebraic version of Tseitin's extension rule, which allows to introduce new variables. We will denote our generalization of $\mathsf{Polynomial}$ $\mathsf{Calculus}$ as $\mathsf{Ext}$-$\mathsf{PC}^{\surd}$.

In this work we give a positive answer to the question raised in \cite{IMP19_new} asking for a technique
for proving size lower bounds on Polynomial Calculus without proving any degree lower bounds. 
Also we give an answer to another question raised in \cite{IMP19_new} 
by proving an exponential lower bound for the system with an extension rule even stronger than that in $\Sigma\Pi\Sigma$-$\mathsf{PC}$, 
which is another extension of Polynomial Calculus presented in the aforementioned work.

We consider the following subset-sum instance, called $\mathsf{Binary}$ $\mathsf{Value}$ $\mathsf{Principle}$ ($\mathsf{BVP}_n$) \cite{AGHT19_new,PT18_new}:
$$
1 + x_1 + 2 x_2 + \ldots 2^{n - 1} x_n = 0,
$$
and prove exponential lower bound for the size of $\mathsf{Ext}$-$\mathsf{PC}_{\mathbb{Q}}^{\surd}$ refutations of $\mathsf{BVP}_n$. Note that $\mathsf{Binary}$ $\mathsf{Value}$ $\mathsf{Principle}$ does not correspond to the translation of any CNF formula and thus the question about proving size lower bound on the refutation of formulas in CNF without proving degree lower bounds \textbf{remains open}.

\begin{theorem}
Any $\mathsf{Ext}$-$\mathsf{PC}_{\mathbb{Q}}^{\surd}$ refutation of $\mathsf{BVP}_n$  requires size $2^{\Omega(n)}$.
\end{theorem}

The technique we use for proving this lower bound is similar to the technique for proving conditional $\mathsf{IPS}$ lower bound in \cite{AGHT19_new}.
However, since $\mathsf{Ext}$-$\mathsf{PC}$ proof system is weaker than $\mathsf{Ideal}$ $\mathsf{Proof}$ $\mathsf{System}$, 
we get an unconditional lower bound. The main idea of conditional lower bound in \cite{AGHT19_new} is to prove complexity lower bound on the free term in the end of $\mathsf{IPS}$-refutation of $\mathsf{BVP}_n$ over $\mathbb{Z}$ and then show that $\mathsf{IPS}_\mathbb{Z}$ simulates $\mathsf{IPS}_{\mathbb{Q}}$.
One difference is that instead of concentrating on the \emph{complexity} of computing the free term of the proof, we concentrate on \emph{prime numbers}
being mentioned in the proof (and thus appearing as factors of the free term).

Then we consider $\mathsf{Res}$-$\mathsf{Lin}$ and show that $\mathsf{Ext}$-$\mathsf{PC}_{\mathbb{Q}}^{\surd}$ simulates $\mathsf{Res}$-$\mathsf{Lin}$ 
and thus get an alternative lower bound for $\mathsf{Res}$-$\mathsf{Lin}$.

\begin{corollary}[Informal]
Any $\mathsf{Res}$-$\mathsf{Lin}$ refutation of $\mathsf{BVP}_n$ requires size $2^{\Omega(n)}$.
\end{corollary}

Note that while Part and Tzameret \cite{PT18_new} prove an exponential lower bound on the number of lines in the proof, we prove a bound on the proof size (essentially, on the bit size of scalars appearing in the proof).

\subsection{Organization of the paper}
In Section~\ref{sec:prelim} we recall the definition of Polynomial Calculus ($\mathsf{PC}$) and give the definitions of Polynomial Calculus with square root ($\mathsf{PC}^{\surd}$) and Extended Polynomial Calculus with square root ($\mathsf{Ext}$-$\mathsf{PC}^{\surd}$).

In Section~\ref{section-3} we prove exponential lower bound on the size of $\mathsf{Ext}$-$\mathsf{PC}_{\mathbb{Q}}^{\surd}$ refutations of $\mathsf{BVP}_n$. We start with considering derivations with integer coefficients ($\mathsf{Ext}$-$\mathsf{PC}_{\mathbb{Z}}^{\surd}$) and show that the free term in the end of such refutation of $\mathsf{BVP}_n$ is not just large but also is divisible by all primes less then $2^n$ (see Theorem~\ref{lower bound for integers depth-inf}). Then, in Theorem~\ref{lower bound q}, we convert proofs over $\mathbb{Q}$ into proofs over $\mathbb{Z}$ without changing the set of primes mentioned in the proof and thus get an $\mathsf{Ext}$-$\mathsf{PC}_{\mathbb{Q}}^{\surd}$ lower bound.

In Section~\ref{section-4} we show that $\mathsf{Ext}$-$\mathsf{PC}_{\mathbb{Q}}^{\surd}$ simulates $\mathsf{Res}$-$\mathsf{Lin}$ and thus we get an alternative lower bound for the size of $\mathsf{Res}$-$\mathsf{Lin}$ refutations of $\mathsf{BVP}_n$.

\section{Preliminaries}\label{sec:prelim}
In this paper we are going to work with polynomials over integers or rationals. 
We define the size of a polynomial roughly as the total length of the bit representation of its coefficients:
\begin{definition} [Size of a polynomial]\label{def:PolynomialSize}
Let $f$ be an arbitrary integer or rational polynomial in variables $\{x_1, \ldots, x_n\}$. 
\begin{itemize}
    \item If $f \in \mathbb{Z}[x_1, \ldots, x_n]$ then $Size(f) = \sum \lceil \log |a_i| \rceil$ where $a_i$ are the coefficients of $f$.
    \item If $f \in \mathbb{Q}[x_1, \ldots, x_n]$ then $Size(f) = \sum \lceil \log |p_i| \rceil + \lceil \log|q_i| \rceil $ where $p_i \in \mathbb{Z}$, $q_i \in \mathbb{N}$ and $\frac{p_i}{q_i}$ are the coefficients of $f$.
\end{itemize}
\end{definition}

\begin{definition}[Polynomial Calculus]\label{def:PC}
Let $\Gamma = \{P_1, \ldots, P_m\} \subset \mathbb{F}[x_1, \ldots, x_n]$ be a set of polynomials in variables $\{x_1, \ldots, x_n\}$ over a field $\mathbb{F}$ such that the system of equations $P_1 = 0, \ldots, P_m = 0$ has no solution. A Polynomial Calculus refutation of $\Gamma$ is a sequence of polynomials $R_1, \ldots, R_s$ where $R_s = 1$  and for every $l$ in $\{1, \ldots, s\}$, $R_l \in \Gamma$ or is obtained through one of the following derivation rules for $j, k < l$
\begin{itemize}
    \item $R_l = \alpha R_j + \beta R_k$ for $\alpha, \beta \in \mathbb{F}$
    \item $R_l = x_i R_k$
\end{itemize}
The size of the refutation is $\sum_{l = 1}^s Size(R_l)$. The degree of the refutation is $\max_l deg(R_l)$.

\end{definition}

Now we consider a variant of Polynomial Calculus proof system with additional \textbf{square root derivation rule} (see \cite{GH03}). Moreover, we extend our definition from fields to \textbf{rings}.

\begin{definition}[Polynomial Calculus with square root]\label{def:PCS}
Let $\Gamma = \{P_1, \ldots, P_m\} \subset R[x_1, \ldots, x_n]$ be a set of polynomials in variables $\{x_1, \ldots, x_n\}$ over a ring $R$ such that the system of equations $P_1 = 0, \ldots, P_m = 0$ has no solution. A $\mathsf{PC}^{\surd}_{R}$ refutation of $\Gamma$ is a sequence of polynomials $R_1, \ldots, R_s$ where $R_s = M$ for some constant $M \in R, M \neq 0$ and for every $l$ in $\{1, \ldots, s\}$, $R_l \in \Gamma$ or is obtained through one of the following derivation rules for $j, k < l$
\begin{itemize}
    \item $R_l = \alpha R_j + \beta R_k$ for $\alpha, \beta \in R$
    \item $R_l = x_i R_k$ for some $i \in \{1, \ldots, n\}$
    \item $R_l^2 = R_k$ for some $i \in \{1, \ldots, n\}$
\end{itemize}
The size of the refutation is $\sum_{l = 1}^s Size(R_l)$, where $Size(R_l)$ is the size of the polynomial $R_l$. The degree of the refutation is $\max_l deg(R_l)$.
\end{definition}

\begin{note}
We will consider $\mathbb{Q}$ or $\mathbb{Z}$ as the ring $R$. For both of those rings, if we consider \textbf{Boolean} case, where axioms $x_i^2 - x_i = 0$ added, our system will be complete, which means that for every unsatisfiable over $\{0, 1\}$ assignment system $\{f_i(\vec x) = 0\}$  there is a  $\mathsf{PC}^{\surd}_{R}$ refutation. Also, note that if $R$ is a domain and $P^2 = 0$ for some $P \in R[\vec x]$, then $P = 0$.
\end{note}

We now define a variant of $\mathsf{PC}^{\surd}_{R}$, $\mathsf{Ext}$-$\mathsf{PC}^{\surd}_{R}$ where the proof system is additionally allowed to introduce new variables $y_i$ corresponding to arbitrary polynomials in the original variables $x_i$.

\begin{definition}[Extended Polynomial Calculus with square root]\label{def:Depth-inf-PC}
Let $\Gamma = \{P_1, \ldots, P_m\} \subset R[x_1, \ldots, x_n]$ be a set of polynomials in variables $\{x_1, \ldots, x_n\}$ over a ring $R$ such that the system of equations $P_1 = 0, \ldots, P_m = 0$ has no solution. A $\mathsf{Ext}$-$\mathsf{PC}^{\surd}_{R}$ refutation of $\Gamma$ is a $\mathsf{PC}^{\surd}_{R}$ refutation of a set 
$$\Gamma' = \{P_1, \ldots, P_m, y_1 - Q_1(x_1, \ldots, x_n), y_2 - Q_2(x_1, \ldots, x_n, y_1), \ldots, y_m - Q_m(x_1, \ldots, x_n, y_1, \ldots, y_{m - 1})\}
$$ 
where $Q_i \in R[\vec{x}, y_1, \ldots, y_{i - 1}]$ are arbitrary polynomials. 

The size of the $\mathsf{Ext}$-$\mathsf{PC}^{\surd}_{R}$ refutation is equal to the size of the $\mathsf{PC}^{\surd}_{R}$ refutation of $\Gamma'$.

\end{definition}

\section{Lower bound}\label{section-3}
In order to prove lower bound for the $\mathsf{Ext}$-$\mathsf{PC}^{\surd}_{\mathbb{Q}}$ proof system, we consider the following subset-sum instance \cite{AGHT19_new, PT18_new}:

\begin{definition}[Binary Value Principle $\mathsf{BVP}_n$]
The \textbf{binary value principle} over the variables $x_1,\dots, x_n$, $\mathsf{BVP}_n$  for short, is the following unsatisfiable system of linear equations:
$$
x_1 + 2 x_2 + \ldots 2^{n - 1} x_n + 1 = 0,
$$
$$
x_1^2 - x_1 = 0, \; x_2^2 - x_2 = 0, \; \ldots, \; x_n^2 - x_n = 0.
$$
\end{definition}

\begin{theorem}
\label{lower bound for integers depth-inf}
Any $\mathsf{Ext}$-$\mathsf{PC}^{\surd}_{\mathbb{Z}}$ refutation of $\mathsf{BVP}_n$  requires size $\Omega(2^n)$. Moreover, the absolute value of the constant in the end of our $\mathsf{Ext}$-$\mathsf{PC}^{\surd}_{\mathbb{Z}}$ refutation consists of at least $C \cdot 2^n$ bits for some constant $C > 0$. Also, the constant in the end of our $\mathsf{Ext}$-$\mathsf{PC}^{\surd}_{\mathbb{Z}}$ refutation is divisible by every prime number less than $2^n$.
\end{theorem}

\begin{proof}

Assume that $\{R_1, \ldots, R_t\}$ is the $\mathsf{Ext}$-$\mathsf{PC}^{\surd}_{\mathbb{Z}}$ refutation of $\mathsf{BVP}_n$. Then we know that $\{R_1, \ldots, R_t\}$ is $\mathsf{PC}_{\mathbb{Z}}^{\surd}$ refutation of some set 
$$
\Gamma' = \{G(\vec x), F_1(\vec x), \ldots, F_n(\vec x), y_1 - Q_1(\vec x), \ldots y_m - Q_m(\vec x, y_1, \ldots, y_{m - 1})\}
$$
where $G(\vec x) = 1 + \sum_{i = 1}^{i = n} 2^{(i - 1)} x_i$, $F_i(\vec x) = x_i^2 - x_i$ and $Q_i \in \mathbb{Z}[\vec x, y_1, \ldots, y_{i - 1}]$.

By the definition of $\mathsf{Ext}$-$\mathsf{PC}^{\surd}_{\mathbb{Z}}$ refutation we know that there exists an integer constant $M \neq 0$ such that $F_t = M$. 

\begin{claim}\label{cla:M-div-every-prime}
$M$ is divisible by every prime number less than $2^n$.
\end{claim}
\begin{proofclaim}
Consider arbitrary integer number $0 \le k < 2^n$ and its binary representation $b_{1}, \ldots, b_{n}$. Let $k + 1$ be \textbf{prime}. Then $G(b_1, \ldots, b_n) = k + 1$, $F_i(b_1, \ldots, b_n) = b_i^2 - b_i = 0$. Also consider integers $c_1, \ldots, c_m$ such that $c_i = Q_i(b_1, \ldots, b_n, c_1, c_2, \ldots, c_{i - 1})$. Now we will prove by induction that every integer number $R_i(b_1, \ldots, b_n, c_1, \ldots, c_m)$ is divisible by $k + 1$ and thus $M$ is divisible by every prime number less than $2^n$.

\bfseries Base case: \mdseries if $i = 1$, then $R_i = G(b_1, \ldots, b_n, c_1, \ldots, c_m) = k + 1$ or $R_i = F_i(b_1, \ldots, b_n, c_1, \ldots, c_m) = 0$ or $R_i(b_1, \ldots, b_n, c_1, \ldots, c_m) = c_i - Q_i(b_1, \ldots, b_n, c_1, \ldots, c_{i - 1}) = 0$ which means that $R_i$ is divisible by $k + 1$.

\bfseries Induction step: \mdseries suppose we know that $R_j$ is divisible by $k + 1$ for any $j \le i$. Now we will show it for $R_{i + 1}$. There are four cases:
\begin{enumerate}
    \item If $R_{i + 1} \in \Gamma'$, then this case is equivalent to the base case and $R_{i + 1}(b_1, \ldots, b_n, c_1, \ldots, c_m)$ is divisible by $k + 1$.
    \item If $R_{i + 1} = \alpha R_j + \beta R_s$ for $\alpha, \beta \in \mathbb{Z}$ and $j, s \le i$, then 
    $R_{i + 1}(b_1, \ldots, b_n, c_1, \ldots, c_m)$ is divisible by $k + 1$ because $R_j(b_1, \ldots, b_n, c_1, \ldots, c_m)$ and $R_s(b_1, \ldots, b_n, c_1, \ldots, c_m)$ are divisible by $k + 1$ and $\alpha$ and $\beta$ are integers.
    \item If $R_{i + 1} = x_j R_s$ or $R_{i + 1} = y_j R_s$, then $R_{i + 1}(b_1, \ldots, b_n, c_1, \ldots, c_m)$ is divisible by $k + 1$ because $R_s(b_1, \ldots, b_n, c_1, \ldots, c_m)$ is divisible by $k + 1$ and $b_i$ and $c_i$ are integers. 
    \item If $R_{i + 1}^2 = R_s$, then we know that $R_s(b_1, \ldots, b_n, c_1, \ldots, c_m)$ is divisible by $k + 1$. Suppose $R_{i + 1}(b_1, \ldots, b_n, c_1, \ldots, c_m)$ is not divisible by $k + 1$. Then $R_{i + 1}(b_1, \ldots, b_n, c_1, \ldots, c_m)^2$ is not divisible by $k + 1$ since $k + 1$ is \textbf{prime}. But $R_{i + 1}(b_1, \ldots, b_n, c_1, \ldots, c_m)^2 = R_s(b_1, \ldots, b_n, c_1, \ldots, c_m)$ which leads us to a contradiction. 
\end{enumerate}
Since every $R_{i}(b_1, \ldots, b_n, c_1, \ldots, c_m)$ is divisible by $k + 1$, we know that $M = R_{s}(b_1, \ldots, b_n, c_1, \ldots, c_m)$ is divisible by every $k + 1$ less than $2^n$, and in particular  $M$ is divisible by every prime number less than $2^n$.
\end{proofclaim}

So we know that $M$ is divisible by the product of all prime numbers less than $2^n$. Then we know that $|M| > (\pi(2^n))!$ where $\pi(2^n)$ is the number of all prime numbers less than $2^n$. By the prime number theorem $\pi(2^n) > C \frac{2^n}{n}$. By Stirling's approximation we get
$$
|M| > \left(C \frac{2^n}{n}\right)! > C' \cdot \left(C  \frac{2^n}{n}\right)^{C  \frac{2^n}{n}} > C'' \left(2^{\frac{n}{2}}\right)^{C  \frac{2^n}{n}} > C'' 2^{(2^n C_0)}
$$
which means that $M$ consists of at least $C_1 \cdot 2^n$ bits and therefore any $\mathsf{Ext}$-$\mathsf{PC}^{\surd}_{\mathbb{Z}}$ refutation of $\mathsf{BVP}_n$  requires size $\Omega(2^n)$. 

\end{proof}

In order to prove a lower bound over $\mathbb{Q}$, we need to convert an $\mathsf{Ext}$-$\mathsf{PC}^{\surd}_{\mathbb{Q}}$ proof into $\mathsf{Ext}$-$\mathsf{PC}^{\surd}_{\mathbb{Z}}$ proof.

\begin{theorem}\label{lower bound q}
Any $\mathsf{Ext}$-$\mathsf{PC}^{\surd}_{\mathbb{Q}}$ refutation of $\mathsf{BVP}_n$  requires size $\Omega(2^n)$.
\end{theorem} 

\begin{proof}
Assume that $\{R_1, \ldots, R_t\}$ is the $\mathsf{Ext}$-$\mathsf{PC}^{\surd}_{\mathbb{Q}}$ refutation of $\Gamma$ of the size $S$. Then we know that $\{R_1, \ldots, R_t\}$ is a $\mathsf{PC}_{\mathbb{Q}}^{\surd}$ refutation of some set $\Gamma' = \Gamma \cup \{y_1 - Q_1(\vec x), \ldots, y_m - Q_m(\vec x, y_1, \ldots, y_{m - 1})\}$ where  $Q_i \in \mathbb{Q}[\vec x, \vec y]$. Also, we know that $R_t = M$ for some $M \in \mathbb{Q}$.

Consider integers $M_1, \ldots, M_m$ where $M_i$ is equal to the product of denominators of all coefficients of polynomial $Q_i$. Also consider all polynomials $R_j(\vec x, \vec y)$ which was derived by using linear combination rule which means that $R_j = \alpha R_i + \beta R_k$. Then we consider \textbf{all} constants $\alpha$ and $\beta$ occurring in linear combination derivations in our proof. Let's denote the set of those constants as $\{\gamma_1, \gamma_2, \ldots, \gamma_f\} \subset \mathbb{Q}$. Now consider the set of all \textbf{denominators} of the constants in $\{\gamma_1, \gamma_2, \ldots, \gamma_f\}$ and denote this set as $\{\delta_1, \delta_2, \ldots, \delta_l\} \subset \mathbb{N}$. 

Also consider the products of all denominators of coefficients of polynomials $\{R_1, \ldots, R_t\}$. We will denote the set of those integers as $\{L_1, \ldots, L_t\} \subset \mathbb{N}$.

Now we will construct the $\mathsf{Ext}$-$\mathsf{PC}^{\surd}_{\mathbb{Z}}$ refutation of $\Gamma$ such that the constant in the end of this proof is equal to $M_1^{c_1} \cdot M_2^{c_2} \cdots M_m^{c_m} \cdot \delta_1^{c_{m + 1}} \cdots \delta_l^{c_{m + l}} \cdot  L_1^{c_{m + l + 1}} \cdots L_t^{c_{m + l + t}}  \cdot M$ where $\{c_1, c_2, \cdots, c_{m + l + t}\} \subset \mathbb{N} \cup \{0\}$.
 
Firstly, we will translate polynomials $Q_i$ into some integer polynomials $Q_i'$. Consider $Q_1'(\vec x) = M_1 \cdot Q_1(\vec x)$ where $M_1$ is equal to the product of denominators of all coefficients of polynomial $Q_1$. Then $Q_1' \in \mathbb{Z}[\vec x]$ and $T_1 = M_1$. Then consider $Q_2' (\vec x, y_1') = T_2 \cdot Q_2 (\vec x, \frac{y_1'}{T_1})$ where $T_2$ is equal to $T_1^{\alpha_{1 1}} \cdot M_2$ where $\alpha_{1 1}$ is an \textbf{arbitrary} non-negative integer such that $Q_2' \in \mathbb{Z}[\vec x, y_1']$. Then for every $i$ we consider $Q_i'(\vec x, y_1', \ldots, y_{i - 1}') = T_i \cdot Q_{i}(\vec x, \frac{y_1'}{T_1}, \ldots, \frac{y_{i -1 }'}{T_{i - 1}})$ where $T_i = T_1^{\alpha_{i 1}} \cdot T_2^{\alpha_{i 2}} \cdots T_{i - 1}^{\alpha_{i i - 1}} \cdot M_i$ where $\alpha_{i 1}, \ldots, \alpha_{i i - 1}$ are \textbf{arbitrary} integers such that $Q_{i}' \in \mathbb{Z}[\vec x, y_1', \ldots, y_{i - 1}']$. Note that we are not interested in the size of the integers $\alpha_{i j}$ so they could be arbitrary large.

Now we will construct $\mathsf{PC}_{\mathbb{Q}}^{\surd}$ refutation $\{R_1', \ldots, R_s'\}$ of the set $\Gamma'' = \Gamma \cup \{y_1' - Q_1'(\vec x), \ldots y_m' - Q_m'(\vec x, y_1', \ldots, y_{m - 1}')\}$ of the following form: this refutation duplicates the original refutation $\{R_1, \ldots, R_t\}$ in all cases except when the polynomial $R_i$ was derived by multiplying by some variable $y_j$ from some polynomial $R_k$. In this case we will multiply corresponding polynomial by $y_j'$ and then multiply it by $\frac{1}{T_j}$.  

Formally, we will prove the following claim:
\begin{claim}\label{cla:Q-Z trans}
There is an $\mathsf{PC}_{\mathbb{Q}}^{\surd}$ refutation $\{R_1', \ldots, R_s'\}$ of the set $\Gamma'' = \Gamma \cup \{y_1' - Q_1'(\vec x), \ldots y_m' - Q_m'(\vec x, y_1', \ldots, y_{m - 1}')\}$ for which the following properties holds:
\begin{itemize}
    \item For every polynomial $R_i'(\vec x, y_1', \ldots, y_m')$ one of the following equations holds:
    $R_i'(\vec x, y_1 \cdot T_1, \ldots, y_m \cdot T_m) = R_j(\vec x, y_1, \ldots, y_m)$  for some $j$ or $R_i'(\vec x, y_1 \cdot T_1, \ldots, y_m \cdot T_m) =  T_k \cdot R_j(\vec x, y_1, \ldots, y_m)$ for some $k$ and $j$.
    \item If $R_i'(\vec x, y_1', \ldots, y_m')$ was derived from $R_j'(\vec x, y_1', \ldots, y_m')$ and $R_k'(\vec x, y_1, \ldots, y_m)$ by taking linear combination with rational constants $\alpha$ and $\beta$ (which means that $R_i' = \alpha R_j' + \beta R_k'$), then $\alpha = \frac{1}{T_f}$ and $\beta = 0$ for some $f$ or there is some polynomial $R_h(\vec x, y_1', \ldots, y_m')$ which was derived from some polynomials $R_k$ and $R_l$ by using linear combination with constants $\alpha$ and $\beta$.
\end{itemize}

\end{claim}

\begin{proofclaim}
The proof is an easy (but lengthy) inductive argument and is given in the \hyperref[appendix]{Appendix}.
\end{proofclaim}

Now we will show that $\Gamma''$ has a $\mathsf{PC}_{\mathbb{Z}}^{\surd}$ refutation in which the constant in the end is equal to 
$$
M_1^{c_1} \cdot M_2^{c_2} \cdots M_m^{c_m} \cdot \delta_1^{c_{m + 1}} \cdots \delta_l^{c_{m + l}} \cdot  L_1^{c_{m + l + 1}} \cdots L_t^{c_{m + l + t}}  \cdot M.
$$
In order to do this we will fix a $\mathsf{PC}_{\mathbb{Q}}^{\surd}$ refutation $\{R_1', \ldots, R_s'\}$ of $\Gamma''$ with the properties from the \hyperref[cla:Q-Z trans]{Claim 3.4} and construct a $\mathsf{PC}_{\mathbb{Z}}^{\surd}$ refutation of $\Gamma''$ by induction. Moreover, we will construct a $\mathsf{PC}_{\mathbb{Z}}^{\surd}$ refutation $\{R_1'', \ldots, R_f''\}$ in which every polynomial $R_i''$ is equal to $M_1^{d_1} \cdot M_2^{d_2} \cdots M_m^{d_m} \cdot \delta_1^{d_{m + 1}} \cdots \delta_l^{d_{m + l}} \cdot  L_1^{d_{m + l + 1}} \cdots L_t^{d_{m + l + t}} \cdot R_i'$ for some non-negative integers $d_1, \ldots, d_{m + l + t}$ and some polynomial $R_i'$. 

Informally, we are going to multiply each line in our  $\mathsf{PC}_{\mathbb{Q}}^{\surd}$ refutation by some constant in order to get correct $\mathsf{PC}_{\mathbb{Z}}^{\surd}$ refutation. But since we can't divide polynomials in our $\mathsf{PC}_{\mathbb{Z}}^{\surd}$ refutation by any constant, we will duplicate original $\mathsf{PC}_{\mathbb{Q}}^{\surd}$ refutation multiplied by some constant of the form $M_1^{d_1} \cdot M_2^{d_2} \cdots M_m^{d_m} \cdot \delta_1^{d_{m + 1}} \cdots \delta_l^{d_{m + l}} \cdot  L_1^{d_{m + l + 1}} \cdots L_t^{d_{m + l + t}}$  every time we would like to simulate derivation in the original proof. 

\noindent\textbf{Induction statement:} Let $\{R_1', \ldots, R_i'\}$ be a $\mathsf{PC}_{\mathbb{Q}}^{\surd}$ derivation from $\Gamma''$ with the properties from the \hyperref[cla:Q-Z trans]{Claim 3.4}. Then there exists a $\mathsf{PC}_{\mathbb{Z}}^{\surd}$ derivation $\{R_1'', \ldots, R_f''\}$ from $\Gamma''$ such that 
\begin{itemize}
    \item $f \le 2 i^2$.
    \item There is some constant $F_i = M_1^{b_1} \cdot M_2^{b_2} \cdots M_m^{b_m} \cdot \delta_1^{b_{m + 1}} \cdots \delta_l^{b_{m + l}} \cdot  L_1^{b_{m + l + 1}} \cdots L_t^{b_{m + l + t}} \in \mathbb{N}$ such that 
    $$
    F_i \cdot R_1' = R_{f - i + 1}'', \; F_i \cdot R_2' = R_{f - i + 2}'', \: \ldots, \;  F_i \cdot R_i' = R_{f}''
    $$
\end{itemize}
\mbox{}\\ \noindent{\textit{Base case: }} If $i = 1$ then $R_i' \in \Gamma''$. Then we can take $R_{1}'' = R_i'$.

\mbox{}\\ \noindent{\textit{Induction step: }} Suppose we have already constructed the $\mathsf{PC}_{\mathbb{Z}}^{\surd}$ refutation $\{R_1'', R_2'', \ldots, R_{f}''\}$ for which the induction statement is true. Then there are four cases depending on the way the $R_{i + 1}'$ is derived.

\noindent{\textbf{Case 1}:} If $R_{i + 1}' \in \Gamma''$ then $F_{i + 1} = F_i$ and
$$
R_{f + 1}'' = R_{i + 1}', \:R_{f + 2}'' = F_{i + 1} \cdot R_1', \; R_{f + 3}'' = F_{i + 1} \cdot R_2', \: \ldots, \;  R_{f + i + 1}'' = F_{i + 1} \cdot R_i', \:, R_{f + i + 2}'' = F_{i + 1} \cdot R_{i + 1}'  
$$

\noindent{\textbf{Case 2}:} If $R_{i + 1}' = x_j R_l'$ or $R_{i + 1}' = y_j' R_l'$ then $F_{i + 1} = F_i$, 
$$
R_{f + 1}'' = F_{i + 1} \cdot R_1', \; R_{f + 2}'' = F_{i + 1} \cdot R_2', \: \ldots, \;  R_{f + i}'' = F_{i + 1} \cdot R_i'
$$
and  $R_{f + i + 1}'' = x_j R_{f - i + l}'' = F_{i + 1} \cdot R_{i + 1}'$ or $R_{f + i + 1}'' = y_j R_{f - i + l}'' = F_{i + 1} \cdot R_{i + 1}''$.

\noindent{\textbf{Case 3}:} If $R_{i + 1} = \alpha R_j + \beta R_k$ where $\alpha = \frac{p_1}{q_1}$ and $\beta = \frac{p_2}{q_2}$ where $\{p_1, q_1, p_2, q_2\} \subset \mathbb{Z}$. Then we can take $F_{i + 1} = q_1 q_2 F_i$, 
$$
R_{f + 1}'' = q_1 q_2 \cdot R_{f - i + 1}'' = F_{i + 1} \cdot R_1', \: R_{f + 2}'' = q_1 q_2 \cdot R_{f - i + 2}'' = F_{i + 1} \cdot R_2', \: \ldots, \: R_{f + i}'' = q_1 q_2 \cdot R_{f}'' = F_{i + 1} R_{i}'  
$$
and $R_{f + i + 1}'' = p_1 q_2 \cdot R_{f - i + j}'' + p_2 q_1 \cdot R_{f - i + k}'' = M_{i + 1} R_{i + 1}'$. From the \hyperref[cla:Q-Z trans]{Claim 3.4} we know that $\alpha = \frac{1}{T_k}$ for some $k$ and $\beta = 0$, or $q_2$ and $q_1$ are equal to some $\delta_k$ and $\delta_r$. From the induction statement we know that 
$$
F_i = M_1^{b_1} \cdot M_2^{b_2} \cdots M_m^{b_m} \cdot \delta_1^{b_{m + 1}} \cdots \delta_l^{b_{m + l}} \cdot  L_1^{b_{m + l + 1}} \cdots L_t^{b_{m + l + t}}.
$$
Then, since $T_k = M_1^{r_{1 k}} \cdots M_m^{r_{m k}}$, we know that 
$$
F_{i + 1} = M_1^{b_1'} \cdot M_2^{b_2'} \cdots M_m^{b_m'} \cdot \delta_1^{b_{m + 1}'} \cdots \delta_l^{b_{m + l}'} \cdot  L_1^{b_{m + l + 1}'} \cdots L_t^{b_{m + l + t}'},
$$
and the induction statement stays true.

\noindent{\textbf{Case 4}:} Suppose $R_{i + 1}'^2 = R_j'$. We know that $R_{i  +1}'(x_1, \ldots, x_n, y_1', \ldots,y_m') = R_{k}(x_1, \ldots, x_n, \frac{y_1'}{T_1}, \ldots, \frac{y_m'}{T_m})$ or $R_{i  +1}'(x_1, \ldots, x_n, y_1', \ldots,y_m') = T_h \cdot R_{k}(x_1, \ldots, x_n, \frac{y_1'}{T_1}, \ldots, \frac{y_m'}{T_m})$ for some $h$. Then we can take $M' = L_k \cdot T_1^{\alpha_1} \cdot T_2^{\alpha_2} \cdots T_m^{\alpha_m} = L_k \cdot M_1^{\alpha_1'} \cdot M_2^{\alpha_2'} \cdots M_m^{\alpha_m'}$ for some non-negative integers $\alpha_1, \ldots, \alpha_m$, such that $M' \cdot R_{i + 1}'$ is an integer polynomial. We know that such integers $\alpha_1, \ldots, \alpha_m$ exist since $L_k$ is the product of all denominators of coefficients of polynomial $R_k$.

Then we can take $F_{i + 1} = M' \cdot F_i$. It's obvious that $F_{i + 1} \cdot R_{i + 1}'$ is an integer polynomial. Then we can make the following $\mathsf{PC}_{\mathbb{Z}}^{\surd}$ derivation:
\begin{multline*}
    R_{f + 1}'' = F_i (M')^2 \cdot R_{f - i + j}'' =  (F_i M')^2 \cdot R_j', \\  R_{f + 2}' = M' \cdot R_{f - i + 1}' = F_{i + 1} \cdot R_1, \: R_{f + 3}' = M' \cdot R_{f - i + 2}' = F_{i + 1} \cdot R_2, \: \ldots, \: R_{f + i + 1}' = M' \cdot R_{f}' = F_{i + 1} R_{i}.
\end{multline*}
Then we can take $R_{f + i + 2}'' = F_i M' \cdot R_{i + 1}'$ and since $R_{f + 1}'' = (F_i M')^2 \cdot R_j'$ we know that $(R_{f + i + 2}'')^2 = R_{f + 1}''$ and we get a correct $\mathsf{PC}_{\mathbb{Z}}^{\surd}$ derivation. 

Since  $M' = L_p \cdot M_1^{\alpha_1'} \cdot M_2^{\alpha_2'} \cdots M_m^{\alpha_m'}$ we know that 
$$
F_{i + 1} = M_1^{b_1'} \cdot M_2^{b_2'} \cdots M_m^{b_m'} \cdot \delta_1^{b_{m + 1}'} \cdots \delta_f^{b_{m + l}'} \cdot  L_1^{b_{m + l + 1}'} \cdots L_t^{b_{m + l + t}'},
$$
and the induction statement stays true.

So now we have a $\mathsf{Ext}$-$\mathsf{PC}^{\surd}_{\mathbb{Z}}$ refutation of $\Gamma$ such that the constant in the end of this refutation is equal to $M_1^{c_1} \cdot M_2^{c_2} \cdots M_m^{c_m} \cdot \delta_1^{c_{m + 1}} \cdots \delta_l^{c_{m + l}} \cdot  L_1^{c_{m + l + 1}} \cdots L_t^{c_{m + l + t}}  \cdot M$. Suppose that $M = \frac{p'}{q'}$ where $p \in \mathbb{Z}$ and $q \in \mathbb{N}$. Then, from  \autoref{lower bound for integers depth-inf}  we know that  $M_1^{c_1} \cdot M_2^{c_2} \cdots M_m^{c_m} \cdot \delta_1^{c_{m + 1}} \cdots \delta_f^{c_{m + l}} \cdot  L_1^{c_{m + l + 1}} \cdots L_t^{c_{m + l + t}}  \cdot p'$ is divisible by every prime number less than $2^n$. Since $M_1, \ldots, M_m$, $\delta_1, \ldots, \delta_l$, $L_1, \ldots, L_t$ are positive integers we know that 
$
M_1 \cdot M_2 \cdots M_m \cdot \delta_1 \cdots \delta_l \cdot  L_1 \cdots L_t  \cdot p'
$
is divisible by every prime number less than $2^n$. Also we know that 
$$
\log \lceil M_1 \rceil + \cdots + \log \lceil M_m \rceil + \log \lceil \delta_1 \rceil + \cdots + \log \lceil \delta_l \rceil + \log \lceil L_1 \rceil  + \cdots +  \log \lceil L_t \rceil + \log \lceil p \rceil \le O(Size(S))
$$ 
because all constants $M_1, \ldots, M_m, L_1, \ldots, L_t$ are products of denominators in the lines of our refutation $\{R_1, \ldots, R_t\}$ and constants $\delta_1, \ldots, \delta_l$ are denominators of rationals in linear combinations used in our derivation. 

On the other hand, we know that 
$$
M_1 \cdot M_2 \cdots M_m \cdot \delta_1 \cdots \delta_l \cdot  L_1 \cdots L_t  \cdot p' \ge 2^{\Omega(n)}
$$
since our product is divisible by every prime number less than $2^n$. Then we know that $S \ge 2^{\Omega(n)}$.

\end{proof}

\newpage

\section{Connection between $\mathsf{Res}$-$\mathsf{Lin}$, $\mathsf{Ext}$-$\mathsf{PC}_{\mathbb{Q}}^{\surd}$ and $\mathsf{Ext}$-$\mathsf{PC}_{\mathbb{Q}}$}\label{section-4}
Following \cite{RT07}, we define $\mathsf{Res}$-$\mathsf{Lin}$ proof system.
\begin{definition}
A \textbf{disjunction of linear equations} is of the following general form: 
\begin{equation}
   (a_1^{(1)} x_1 + \ldots + a_n^{(1)} x_n = a_0^{(1)}) \vee \cdots \vee (a_1^{(t)} x_1 + \ldots + a_n^{(t)} x_n = a_0^{(t)})  
\end{equation}
where $t \ge 0$ and the coefficients $a_i^{j}$ are \textbf{integers} (for all $0 \le i \le n$, $1 \le j \le t$).  The semantics of such a disjunction is the natural one: We say that an assignment of integral values to the variables $x_1, \ldots, x_n$ satisfies (1) if and only if there exists $j \in \{1, \ldots, t\}$ so that the equation $a_1^{(j)} x_1 + \ldots + a_n^{(j)} x_n = a_0^{(j)}$ holds under the given assignment. 

The \textbf{size} of the disjunction of linear equations is $\sum_{i = 1}^{n} \sum_{j = 1}^{t} |a_i^{(j)}|$ if all coefficients are written in \textbf{unary} notation. If all coefficients are written in \textbf{binary} notation then the \textbf{size} is equal to $\sum_{i = 1}^{n} \sum_{j = 1}^{t} \lceil \log |a_i^{(j)}| \rceil$.
\end{definition}

\begin{definition}
Let $K := \{K_1, \ldots, K_m\}$ be a collection of disjunctions of linear equations. An $\mathsf{Res}$-$\mathsf{Lin}$ proof from $K$ of a disjunction of linear equations $D$ is a finite sequence $\pi = (D_1, \ldots, D_l)$ of disjunctions of linear equations, such that $D_l = D$ and for every $i \in \{1, \ldots, l\}$, either $D_i = K_j$ for some $j \in \{1, \ldots, m\}$, or $D_i$ is a Boolean axiom $(x_h = 0) \vee (x_h = 1)$ for some $h \in \{1, \ldots, n\}$, or $D_i$ was deduced by one of the following $\mathsf{Res}$-$\mathsf{Lin}$ inference rules, using $D_j$, $D_k$ for some $j, k < i$:
\begin{itemize}
    \item \textbf{Resolution}: Let $A, B$ be two, possibly empty, disjunctions of linear equations and let $L_1$, $L_2$ be two linear equations. From $A \vee L_1$ and $B \vee L_2$ derive $A \vee B \vee (\alpha L_1 + \beta L_2)$ where $\alpha, \beta \in \mathbb{Z}$.
    \item \textbf{Weakening}: From a (possibly empty) disjunction of linear equations $A$ derive $A \vee L$, where $L$ is an arbitrary linear equation over $\{x_1, \ldots, x_n\}$.
    \item \textbf{Simplification}: From $A \vee (k = 0)$ derive $A$, where $A$ is a, possibly empty, disjunction of linear equations and $k \neq 0$ is a constant.
    \item \textbf{Contraction}: From $A \vee L \vee L$ derive $A \vee L$, where $A$ is a, possibly empty, disjunction of linear equations and $L$ is some linear equation.
\end{itemize}
Note that we assume that the order of equations in the disjunction is not significant, while we contract identical equations, especially.  

An  $\mathsf{Res}$-$\mathsf{Lin}$ \textbf{refutation} of a collection of disjunctions of linear equations $K$ is a proof of the empty disjunction from $K$. The \textbf{size} of an $\mathsf{Res}$-$\mathsf{Lin}$ proof $\pi$ is the total size of all the disjunctions of linear equations in $\pi$. 

If all coefficients in our $\mathsf{Res}$-$\mathsf{Lin}$ proof $\pi$ are written in the \textbf{unary} notation then we denote this proof an $\mathsf{Res}$-$\mathsf{Lin}_{U}$ derivation. Otherwise, if all coefficients are written in the \textbf{binary} notation then we denote this proof an $\mathsf{Res}$-$\mathsf{Lin}_{B}$ derivation.
\end{definition}

\begin{note}
In the original $\mathsf{Res}$-$\mathsf{Lin}$ proof system duplicate linear equations can be discarded from the disjunction. Instead, we will use \textbf{contraction} rule explicitly. It is easy to see that both these variants of $\mathsf{Res}$-$\mathsf{Lin}$ system are equivalent.
\end{note}

\begin{definition} 
Let $D$ be  a disjunction of linear equations:
\begin{equation*}
   (a_1^{(1)} x_1 + \ldots + a_n^{(1)} x_n = a_0^{(1)}) \vee \cdots \vee (a_1^{(t)} x_1 + \ldots + a_n^{(t)} x_n = a_0^{(t)})  
\end{equation*}
We denote by $\widehat{D} $ its translation into the following system of polynomial equations:
$$
   y_1 \cdot y_2 \cdots y_t = 0 
$$
$$
    y_1 = a_1^{(1)} x_1 + \ldots + a_n^{(1)} x_n - a_0^{(1)}, \; y_2 = a_1^{(2)} x_1 + \ldots + a_n^{(2)} x_n - a_0^{(2)}, \; \ldots, \; y_t = a_1^{(t)} x_1 + \ldots + a_n^{(t)} x_n - a_0^{(t)}
$$
If $D$ is the empty disjunction, we define $\widehat{D} $ to be the single polynomial equation 1 = 0.
\end{definition}

Now we will prove that $\mathsf{Ext}$-$\mathsf{PC}_{\mathbb{Q}}^{\surd}$ p-simulates $\mathsf{Res}$-$\mathsf{Lin}_{B}$ and $\Sigma \Pi \Sigma$-$PC_{\mathbb{Q}}$ p-simulates $\mathsf{Res}$-$\mathsf{Lin}_{U}$.

\begin{theorem}
\label{Res-Lin-b simulation}
Let $\pi = (D_1, \ldots, D_l)$ be an  $\mathsf{Res}$-$\mathsf{Lin}_{B}$ proof sequence of $D_l$ from some collection of initial disjunctions of linear equations $Q_1, \ldots, Q_m$. Also consider $L_1, \ldots, L_t$ --- all affine forms that we have in all disjunctions in our $\mathsf{Res}$-$\mathsf{Lin}_{B}$ proof sequence. 

Then, there exists an $\mathsf{PC}_{\mathbb{Q}}^{\surd}$ proof of $\widehat{D}_l$ from $\widehat{Q}_1 \cup \ldots \cup \widehat{Q}_m \cup \{y_1 = L_1, y_2 = L_2, \ldots, y_t = L_t\}$ of size at most $O(p(Size(\pi)))$ for some polynomial $p$.
\end{theorem}

\begin{proof}
We proceed by induction on the number of lines in $\pi$.
\mbox{}\\ \noindent{\textit{Base case:}} An $\mathsf{Res}$-$\mathsf{Lin}_{B}$ axiom $Q_i$ is translated into $\widehat{Q_i}$ and $\mathsf{Res}$-$\mathsf{Lin}_{B}$ Boolean axiom  $(x_i = 0) \vee (x_i = 1)$ is translated into $\mathsf{PC}$ axiom $x_i^2 - x_i = 0$.
\mbox{}\\ \noindent{\textit{Induction step: }} Now we will simulate all $\mathsf{Res}$-$\mathsf{Lin}_{B}$ derivation rules in the $\mathsf{PC}_{\mathbb{Q}}^{\surd}$ proof.
\begin{itemize}
    \item \textbf{Resolution}: Assume that $D_i = A \vee B \vee (\alpha L_1 + \beta L_2)$ where $D_j = A \vee L_1$ and $D_k = B \vee L_2$. Then, we have already derived polynomial equations 
    $$
    y_{j 1} = (a_{j 1}^{(1)} x_1 + \ldots + a_{j n}^{(1)} x_n - a_{j 0}^{(1)}), \; \ldots, \; y_{j t_j} = (a_{j 1}^{(t_j)} x_1 + \ldots + a_{j n}^{(t_j)} x_n - a_{j 0}^{(t_j)}),
    $$
    $$
    y_{k 1} = (a_{k 1}^{(1)} x_1 + \ldots + a_{k n}^{(1)} x_n - a_{k 0}^{(1)}), \; \ldots, \; y_{k t_k} = (a_{k 1}^{(t_k)} x_1 + \ldots + a_{k n}^{(t_k)} x_n - a_{k 0}^{(t_k)}),
    $$
    $$
    y_{j 1} \cdot y_{j 2} \cdots y_{j t_j} = 0, \; y_{k 1} \cdot y_{k 2} \cdots y_{k t_k} = 0
    $$ 
    where 
    $$
    A = (a_{j 1}^{(2)} x_1 + \ldots + a_{j n}^{(2)} x_n = a_{j 0}^{(2)}) \vee \cdots \vee (a_{j 1}^{(t_j)} x_1 + \ldots + a_{j n}^{(t_j)} x_n = a_{j 0}^{(t_j)}), 
    $$
    $$
    B = (a_{k 1}^{(2)} x_1 + \ldots + a_{k n}^{(2)} x_n = a_{k 0}^{(2)}) \vee \cdots \vee (a_{k 1}^{(t_k)} x_1 + \ldots + a_{k n}^{(t_k)} x_n = a_{k 0}^{(t_k)})
    $$
    $$
    L_1 = (a_{j 1}^{(1)} x_1 + \ldots + a_{j n}^{(1)} x_n = a_{j 0}^{(1)}), \; L_2 = (a_{k 1}^{(1)} x_1 + \ldots + a_{k n}^{(1)} x_n = a_{k 0}^{(1)}).
    $$
    Then we can derive $y_{j 1} \cdot y_{j 2} \cdots y_{j t_j} \cdot y_{k 2} \cdots y_{k t_k} = 0$, $y_{j 1} \cdot y_{j 2} \cdots y_{j t_j} \cdot y_{k 2} \cdots y_{k t_k} = 0$ and thus $(\alpha y_{j 1} + \beta y_{k 1}) \cdot y_{j 2} \cdots y_{j t_j} \cdot y_{k 2} \cdots y_{k t_k} = 0$. Then there is some variable $y_{i 1}$ for which holds $y_{i 1} = \alpha (a_{j 1}^{(1)} x_1 + \ldots + a_{j n}^{(1)} x_n - a_{j 0}^{(1)}) + \beta (a_{k 1}^{(1)} x_1 + \ldots + a_{k n}^{(1)} x_n - a_{k 0}^{(1)})$ and we can derive $y_{i 1} = \alpha y_{j 1} + \beta y_{k 1}$. Then we can derive $y_{i 1} \cdot y_{j 2} \cdots y_{j t_j} \cdot y_{k 2} \cdots y_{k t_k} = 0$ which is  part of $\widehat{D}_i$.
    \item \textbf{Weakening}: Assume that $D_i = D_j \vee L$ where $L$ is a linear equation. Then, we have already derived polynomial equations 
    $$
    y_{j 1} = (a_{j 1}^{(1)} x_1 + \ldots + a_{j n}^{(1)} x_n - a_{j 0}^{(1)}), \; \ldots, \; y_{j t_j} = (a_{j 1}^{(t_j)} x_1 + \ldots + a_{j n}^{(t_j)} x_n - a_{j 0}^{(t_j)}),
    $$
    $$
    y_{j 1} \cdot y_{j 2} \cdots y_{j t_j} = 0.
    $$
    We know that there is some variable $y_0$ for which $y_0 = b_{1}x_1 + \ldots b_n x_n - b_0$ where $L$ is a linear equation $b_{1}x_1 + \ldots b_n x_n = b_0$. From $y_{j 1} \cdot y_{j 2} \cdots y_{j t_j} = 0$ we can derive $y_0 \cdot y_{j 1} \cdot y_{j 2} \cdots y_{j t_j} = 0$ which is part of $\widehat{D}_i$.
    \item \textbf{Simplification}: Suppose that $D_i = A$ and $D_j = A \vee (k = 0)$ where $k \in \mathbb{Z}$, $k \neq 0$. Then, we have already derived polynomial equations
    $$
    y_{j 1} = (a_{j 1}^{(1)} x_1 + \ldots + a_{j n}^{(1)} x_n - a_{j 0}^{(1)}), \; \ldots, \; y_{j t_j - 1} = (a_{j 1}^{(t_j - 1)} x_1 + \ldots + a_{j n}^{(t_j - 1)} x_n - a_{j 0}^{(t_j - 1)}), \; y_{j t_j} = k,
    $$
    $$
    y_{j 1} \cdot y_{j 2} \cdots y_{j t_j} = 0.
    $$
    From equation $y_{j 1} \cdot y_{j 2} \cdots y_{j t_j} = 0$ we can derive equation $y_{j 1} \cdot y_{j 2} \cdots y_{j t_j - 1} \cdot k = 0$ from which we can derive $y_{j 1} \cdot y_{j 2} \cdots y_{j t_j - 1} = 0$ which is part of $\widehat{D}_i$.
    \item \textbf{Contraction}: Assume that $D_i = A \vee L$ and $D_j \vee L \vee L$ where $L$ is a linear equation. 
    Then, we have already derived polynomial equations 
    $$
    y_{j 1} = (a_{j 1}^{(1)} x_1 + \ldots + a_{j n}^{(1)} x_n - a_{j 0}^{(1)}), \; \ldots, \; y_{j t_{j} - 1} = y_{j t_j} = (a_{j 1}^{(t_j)} x_1 + \ldots + a_{j n}^{(t_j)} x_n - a_{j 0}^{(t_j)}),
    $$
    $$
    y_{j 1} \cdot y_{j 2} \cdots y_{j t_j - 1} \cdot y_{j t_j} = 0.
    $$
    Then we can derive $y_{j t_j - 1} = y_{j t_j}$ and $y_{j 1} \cdot y_{j 2} \cdots y_{j t_j - 2} \cdot (y_{j t_j - 1}^2) = 0$. Using multiplication we can derive  $y_{j 1}^2 \cdot y_{j 2}^2 \cdots y_{j t_j - 2}^2 \cdot (y_{j t_j - 1}^2) = 0$ from which we can derive the equation $y_{j 1} \cdot y_{j 2} \cdots y_{j t_j - 1} = 0$ by using the square root rule. This equation is the last part of $\widehat{D}_i$ because other parts were derived earlier. 
\end{itemize}

\end{proof}

\begin{definition}
Let $\Gamma = \{P_1, \ldots, P_m\} \subset \mathbb{F}[x_1, \ldots, x_n]$ be a set of polynomials in variables $\{x_1, \ldots, x_n\}$ over a ring $R$ such that the system of equations $P_1 = 0, \ldots, P_m = 0$ has no solution. A $\Sigma \Pi \Sigma$-$PC_{\mathbb{Q}}$ refutation of $\Gamma$ is a $\mathsf{PC}_{R}$ refutation of a set $\Gamma' = \{P_1, \ldots, P_m, Q_1, \ldots, Q_m\}$ where $Q_i$ are polynomials of the form $Q_i = y_i - (a_{i 0} + \sum_j a_{i j} x_j)$ for some constants $a_{i j} \in R$. 

The size of the $\Sigma \Pi \Sigma$-$PC_{\mathbb{Q}}$ refutation is equal to the size of the  $\mathsf{PC}_{R}$ refutation of $\Gamma'$.
\end{definition}

\begin{theorem}
\label{Res-Lin-u simulation}
Let $\pi = (D_1, \ldots, D_l)$ be an $\mathsf{Res}$-$\mathsf{Lin}_{U}$ proof sequence of $D_l$, from some collection of initial disjunctions of linear equations $Q_1, \ldots, Q_m$. Then, there exists an $\Sigma \Pi \Sigma$-$PC_{\mathbb{Q}}$ proof of $\widehat{D}_l$ from $\widehat{Q}_1 \cup \ldots \cup \widehat{Q}_m$ of size at most $O(p(Size(\pi)))$ for some polynomial $p$.
\end{theorem}
\begin{proof}
To prove this theorem we will use the following lemma from \cite{IMP19_new}:
\begin{lemma}[\cite{IMP19_new}]
Let $\Gamma = \{P_1, \ldots, P_a, Q_1, \ldots, Q_b, X, Y\}$ be a set of polynomials such that 
$$
P_1 = x_1 - (x - 1), \; P_2 = x_2 - (x - 2), \; \ldots, P_a = x_a - (x - a), 
$$
$$
Q_1 = y_1 - (y - 1), \; Q_2 = y_2 - (y - 2), \; \ldots, Q_b = y_b - (y - b), 
$$
$$
X = x \cdot x_1 \cdot x_2 \cdots x_a, \; Y = y \cdot y_1 \cdot y_2 \cdots y_b.
$$
Then we can derive $\Gamma'$ from $\Gamma$ in $\Sigma \Pi \Sigma$-$PC_{\mathbb{Q}}$ with derivation of size $poly(a b)$ where $\Gamma' = \{Z_0, Z_1, \ldots, Z_{a + b}, Z\}$ and
$$
Z_0 = z - (x + y), \; Z_1 = z_1 - (x + y + 1), \; Z_2 =  z_2 - (x + y + 2) , \; \ldots, Z_{a + b} = z_{a + b} - (x + y + a + b),
$$
$$
Z = z \cdot z_1 \cdot z_2 \cdots z_{a + b}.
$$
\end{lemma}

Now we will prove the theorem by induction on lines in $\pi$.
\mbox{}\\ \noindent{\textit{Base case: }} An $\mathsf{Res}$-$\mathsf{Lin}_{B}$ axiom $Q_i$ is translated into $\widehat{Q_i}$ and $\mathsf{Res}$-$\mathsf{Lin}_{B}$ Boolean axiom $(x_i = 0) \vee (x_i = 1)$ is translated into $\mathsf{PC}$ axiom $x_i^2 - x_i = 0$.
\mbox{}\\ \noindent{\textit{Induction step: }} Now we will simulate all $\mathsf{Res}$-$\mathsf{Lin}_{B}$ derivation rules in the $\mathsf{Ext}$-$\mathsf{PC}_{\mathbb{Q}}^{\surd}$ proof.
\begin{itemize}
    \item \textbf{Resolution}, \textbf{Weakening}, \textbf{Simplification} rules simulation is the same as in \autoref{Res-Lin-b simulation}.
    \item \textbf{Contraction}: Assume that $D_i = A \vee L$ and $D_j \vee L \vee L$ where $L$ is a linear equation. 
    Then, we have already derived polynomial equations 
    $$
    y_{j 1} = (a_{j 1}^{(1)} x_1 + \ldots + a_{j n}^{(1)} x_n - a_{j 0}^{(1)}), \; \ldots, \; y_{j t_{j} - 1} = y_{j t_j} = (a_{j 1}^{(t_j)} x_1 + \ldots + a_{j n}^{(t_j)} x_n - a_{j 0}^{(t_j)}),
    $$
    $$
    y_{j 1} \cdot y_{j 2} \cdots y_{j t_j - 1} \cdot y_{j t_j} = 0.
    $$
    Then we can derive $y_{j t_j - 1} = y_{j t_j}$ and $y_{j 1} \cdot y_{j 2} \cdots y_{j t_j - 2} \cdot (y_{j t_j - 1}^2) = 0$. Using lemma we can introduce new variables $\{z_{-M}, \ldots, z_M\}$ and derive 
    $$
    z_{-M} = y_{j t_{j - 1}} + M, \;, z_{-M + 1} = y_{j t_{j - 1}} + M - 1, \ldots, z_{0} = y_{j t_{j - 1}}, \; z_M =  y_{j t_{j - 1}} - M,
    $$
    $$
    z_{-M} \cdot z_{-M + 1} \cdots z_{M - 1} \cdot z_M = 0,
    $$
    where $M = |a_{j 1}^{(t_{j - 1})}| + |a_{j 2}^{(t_{j - 1})}| + \ldots + |a_{j n}^{(t_{j - 1})}|$. Then we can substitute $y_{j t_{j}} - k$ for each $z_{k}$ one by one and get equation
    $$
    f(y_{j t_{j - 1}}) = 0
    $$
    where $f(y_{j t_{j - 1}}) = b_1 \cdot y_{j t_{j - 1}} + b_2 \cdot y_{j t_{j - 1}}^2 + \ldots +  b_{2 M + 1} \cdot y_{j t_{j - 1}}^{2 M + 1} $ is some polynomial from $\mathbb{Z}[y_{j t_{j - 1}}]$ and $b_1 = (M!)^2 \cdot (-1)^{M}$. Then we can derive the following equation by using multiplication rule:
    \begin{multline*}
      y_{j 1} \cdot y_{j 2} \cdots y_{j t_j - 2} \cdot f(y_{j t_{j - 1}}) = b_1 \cdot y_{j 1} \cdot y_{j 2} \cdots y_{j t_j - 2} \cdot y_{j t_j - 1} + \\ + y_{j 1} \cdot y_{j 2} \cdots y_{j t_j - 2} \cdot (y_{j t_j - 1}^2) \cdot (b_2 + b_3 \cdot y_{j t_{j - 1}} + \ldots +  b_{2 M + 1} \cdot y_{j t_{j - 1}}^{2 M - 1}) =  0.
    \end{multline*}
    Now, using the equation  $y_{j 1} \cdot y_{j 2} \cdots y_{j t_j - 2} \cdot (y_{j t_j - 1}^2) = 0$ we can derive $ b_1 \cdot y_{j 1} \cdot y_{j 2} \cdots y_{j t_j - 2} \cdot y_{j t_j - 1} = 0$ and since $b_1 \neq 0$ we can derive $y_{j 1} \cdot y_{j 2} \cdots y_{j t_j - 2} \cdot y_{j t_j - 1} = 0$. This equation is the last part of $\widehat{D}_i$ because other parts were derived earlier.

\end{itemize}

\end{proof}


Now we will show that our lower bound provides an interesting counterpart to a result from  \cite{PT18_new}.

\begin{theorem}[\cite{PT18_new}]
Any $\mathsf{Res}$-$\mathsf{Lin}_{B}$ refutation of $1 + 2x_1 + \ldots + 2^n x_n = 0$ is of the size $2^{\Omega(n)}$.
\end{theorem}

\begin{proof}
From \autoref{lower bound q} we know that any $\mathsf{Ext}$-$\mathsf{PC}_{\mathbb{Q}}^{\surd}$ refutation of $\mathsf{BVP}_n$  requires size $2^{\Omega(n)}$ and thus from \autoref{Res-Lin-b simulation}  we know that there is some polynomial $p$ such that for any  $\mathsf{Res}$-$\mathsf{Lin}_{B}$ refutation of $\mathsf{BVP}_n$ of size $S$ the equation $p(S) \ge C_0 \cdot 2^{C_1 \cdot n}$ holds. Then we know that for some constant $C$ the equation $S \ge 2^{C \cdot n}$ holds.

\end{proof}

\section*{Open Problems}
\begin{enumerate}
    \item \autoref{Res-Lin-b simulation} says that $\mathsf{Ext}$-$\mathsf{PC}_{\mathbb{Q}}^{\surd}$ p-simulates any $\mathsf{Res}$-$\mathsf{Lin}_{B}$ derivation. Is the square root rule necessary, that is, can we p-simulate $\mathsf{Res}$-$\mathsf{Lin}_{B}$ refutation in the $\mathsf{Ext}$-$\mathsf{PC}_{\mathbb{Q}}$ proof system?
    \item A major question is to prove an exponential lower bound on the size of $\Sigma \Pi \Sigma$-$\mathsf{PC}_{\mathbb{Q}}$ refutation of a translation of a formula in CNF.
\end{enumerate}

\section*{Acknowledgement}
I would like to thank Edward A. Hirsch for guidance and useful discussions at various stages of this work. Also I wish to thank Dmitry Itsykson and Dmitry Sokolov for very helpful comments concerning this work.

\small
\bibliographystyle{plain}

\begin{thebibliography}{10}

\bibitem{Ajt94}
Mikl\'{o}s Ajtai.
\newblock The independence of the modulo $p$ counting principles.
\newblock {\em Electronic Colloquium on Computational Complexity, ECCC},
  (Report no.: TR94-014), December 1994.

\bibitem{AGHT19_new}
Yaroslav Alekseev, Dima Grigoriev, Edward~A. Hirsch, and Iddo Tzameret.
\newblock Semi-algebraic proofs, {IPS} lower bounds and the $\tau$-conjecture:
  Can a natural number be negative?
\newblock In {\em \STOC{2020}}, pages 54--67, 2020.

\bibitem{BeameIKPP96}
Paul Beame, Russell Impagliazzo, Jan Kraj{\'{\i}}{\v{c}}ek, Toniann Pitassi,
  and Pavel Pudl{\'a}k.
\newblock Lower bounds on {H}ilbert's {N}ullstellensatz and propositional
  proofs.
\newblock {\em Proc. London Math. Soc. (3)}, 73(1):1--26, 1996.

\bibitem{BussGIP01}
Sam Buss, Dima Grigoriev, Russell Impagliazzo, and Toniann Pitassi.
\newblock Linear gaps between degrees for the polynomial calculus modulo
  distinct primes.
\newblock {\em Journal of Computer and System Sciences}, 62(2):267 -- 289,
  2001.

\bibitem{BussIKPRS96}
Samuel~R. Buss, Russell Impagliazzo, Jan Kraj{\'{\i}}{\v{c}}ek, Pavel
  Pudl{\'{a}}k, Alexander~A. Razborov, and Ji{\v{r}}{\'{\i}} Sgall.
\newblock Proof complexity in algebraic systems and bounded depth {F}rege
  systems with modular counting.
\newblock {\em \ComputationalComplexity}, 6(3):256--298, 1996.

\bibitem{CHVATAL1989455}
V.~Chv{\'a}tal, W.~Cook, and M.~Hartmann.
\newblock On cutting-plane proofs in combinatorial optimization.
\newblock {\em Linear Algebra and its Applications}, 114-115:455 -- 499, 1989.
\newblock Special Issue Dedicated to Alan J. Hoffman.

\bibitem{CEI96}
Matthew Clegg, Jeffery Edmonds, and Russell Impagliazzo.
\newblock Using the {G}roebner basis algorithm to find proofs of
  unsatisfiability.
\newblock In {\em Proceedings of the 28th Annual ACM Symposium on the Theory of
  Computing (Philadelphia, PA, 1996)}, pages 174--183, New York, 1996. ACM.

\bibitem{CookReckhow74b}
Stephen~A. Cook and Robert~A. Reckhow.
\newblock {C}orrections for ``{O}n the lengths of proofs in the propositional
  calculus (preliminary version)''.
\newblock {\em {SIGACT} News}, 6(3):15--22, July 1974.

\bibitem{CookReckhow74a}
Stephen~A. Cook and Robert~A. Reckhow.
\newblock {O}n the lengths of proofs in the propositional calculus (preliminary
  version).
\newblock In {\em \STOC{1974}}, pages 135--148, 1974.
\newblock For corrections see Cook-Reckhow~\cite{CookReckhow74b}.

\bibitem{CR79}
Stephen~A. Cook and Robert~A. Reckhow.
\newblock The relative efficiency of propositional proof systems.
\newblock {\em J. Symb. Log.}, 44(1):36--50, 1979.
\newblock This is a journal-version of Cook-Reckhow~\cite{CookReckhow74a} and
  Reckhow~\cite{Rec76:PhD}.

\bibitem{CCT87}
W.~Cook, C.~R. Coullard, and G.~Turan.
\newblock On the complexity of cutting plane proofs.
\newblock {\em Discrete Applied Mathematics}, 18:25--38, 1987.

\bibitem{GH03}
Dima Grigoriev and Edward~A. Hirsch.
\newblock Algebraic proof systems over formulas.
\newblock {\em Theoret. Comput. Sci.}, 303(1):83--102, 2003.
\newblock Logic and complexity in computer science (Cr\'eteil, 2001).

\bibitem{GP14}
Joshua~A. Grochow and Toniann Pitassi.
\newblock Circuit complexity, proof complexity, and polynomial identity
  testing: The ideal proof system.
\newblock {\em J. {ACM}}, 65(6):37:1--37:59, 2018.

\bibitem{IMP19_new}
Russell Impagliazzo, Sasank Mouli, and Toniann Pitassi.
\newblock The surprising power of constant depth algebraic proofs.
\newblock In {\em Proceedings of the 35th Annual ACM/IEEE Symposium on Logic in
  Computer Science}, LICS '20, page 591–603, New York, NY, USA, 2020.
  Association for Computing Machinery.

\bibitem{IPS99}
Russell Impagliazzo, Pavel Pudl{\'{a}}k, and Ji{\v{r}}{\'{\i}} Sgall.
\newblock Lower bounds for the polynomial calculus and the gr{\"{o}}bner basis
  algorithm.
\newblock {\em Computational Complexity}, 8(2):127--144, 1999.

\bibitem{ITSYKSON2020102722}
Dmitry Itsykson and Dmitry Sokolov.
\newblock Resolution over linear equations modulo two.
\newblock {\em Annals of Pure and Applied Logic}, 171(1):102722, 2020.

\bibitem{Kra98-Discretely}
Jan Kraj{\'{\i}}{\v{c}}ek.
\newblock Discretely ordered modules as a first-order extension of the cutting
  planes proof system.
\newblock {\em The Journal of Symbolic Logic}, 63(4):1582--1596, 1998.

\bibitem{PT18_new}
Fedor Part and Iddo Tzameret.
\newblock Resolution with counting: Different moduli and dag-like lower bounds.
\newblock In {\em 12th Innovations in Theoretical Computer Science Conference,
  {ITCS} 2020, January, 2020, Seattle, WA, {USA}}, 2020.

\bibitem{Pit97}
Toniann Pitassi.
\newblock Algebraic propositional proof systems.
\newblock In {\em Descriptive complexity and finite models (Princeton, NJ,
  1996)}, volume~31 of {\em DIMACS Ser. Discrete Math. Theoret. Comput. Sci.},
  pages 215--244. Amer. Math. Soc., Providence, RI, 1997.

\bibitem{Pit98}
Toniann Pitassi.
\newblock Unsolvable systems of equations and proof complexity.
\newblock In {\em Proceedings of the International Congress of Mathematicians,
  Vol. III (Berlin, 1998)}, number Vol. III, pages 451--460, 1998.

\bibitem{RT07}
Ran Raz and Iddo Tzameret.
\newblock Resolution over linear equations and multilinear proofs.
\newblock {\em Ann. Pure Appl. Logic}, 155(3):194--224, 2008.

\bibitem{MR1691494}
Alexander~A. Razborov.
\newblock Lower bounds for the polynomial calculus.
\newblock {\em Comput. Complexity}, 7(4):291--324, 1998.

\bibitem{Rec76:PhD}
Robert Reckhow.
\newblock {\em On the lengths of proofs in the propositional calculus}.
\newblock PhD thesis, University of Toronto, 1976.
\newblock Technical Report No . 87.

\bibitem{Sok20}
Dmitry Sokolov.
\newblock (semi)algebraic proofs over \{$\pm$ 1\} variables.
\newblock In {\em Proceedings of the 52nd Annual ACM SIGACT Symposium on Theory
  of Computing}, STOC 2020, page 78–90, New York, NY, USA, 2020. Association
  for Computing Machinery.

\bibitem{Tse68}
Grigori Tseitin.
\newblock {\em On the complexity of derivations in propositional calculus}.
\newblock Studies in constructive mathematics and mathematical logic Part II.
  Consultants Bureau, New-York-London, 1968.

\end{thebibliography}

\normalsize

\appendix
\section*{Appendix}\label{appendix}
\begin{claim-3-4}
There is an $\mathsf{PC}_{\mathbb{Q}}^{\surd}$ refutation $\{R_1', \ldots, R_s'\}$ of the set $\Gamma'' = \Gamma \cup \{y_1' - Q_1'(\vec x), \ldots y_m' - Q_m'(\vec x, y_1', \ldots, y_{m - 1}')\}$ for which the following properties holds:
\begin{itemize}
    \item For every polynomial $R_i'(\vec x, y_1', \ldots, y_m')$ one of the following equations holds:
    $R_i'(\vec x, y_1 \cdot T_1, \ldots, y_m \cdot T_m) = R_j(\vec x, y_1, \ldots, y_m)$  for some $j$ or $R_i'(\vec x, y_1 \cdot T_1, \ldots, y_m \cdot T_m) =  T_k \cdot R_j(\vec x, y_1, \ldots, y_m)$ for some $k$ and $j$.
    \item If $R_i'(\vec x, y_1', \ldots, y_m')$ was derived from $R_j'(\vec x, y_1', \ldots, y_m')$ and $R_k'(\vec x, y_1, \ldots, y_m)$ by taking linear combination with rational constants $\alpha$ and $\beta$ (which means that $R_i' = \alpha R_j' + \beta R_k'$), then $\alpha = \frac{1}{T_f}$ and $\beta = 0$ for some $f$ or there is some polynomial $R_h(\vec x, y_1', \ldots, y_m')$ which was derived from some polynomials $R_k$ and $R_l$ by using linear combination with constants $\alpha$ and $\beta$.
\end{itemize}

\end{claim-3-4}

\begin{proofclaim}
We will construct $\mathsf{PC}_{\mathbb{Q}}^{\surd}$ refutation $\{R_1', R_2', \ldots, R_{s}'\}$ of the set $\Gamma''$ by induction. \smallskip

\noindent\textbf{Induction statement:} 
Let $\{R_1, \ldots, R_i\}$ be a $\mathsf{PC}_{\mathbb{Q}}^{\surd}$ derivation from $\Gamma'$. Then there exists a $\mathsf{PC}_{\mathbb{Q}}^{\surd}$ derivation $\{R_1', \ldots, R_p'\}$ from $\Gamma''$ such that 
\begin{itemize}
    \item $p \le 2 i$.
    \item For every $R_j(x_1, \ldots, x_n, y_1, \ldots, y_m)$ there exists some $R_k'(x_1, \ldots, x_n, y_1', \ldots, y_m')$ such that 
    $$
    R_k'(x_1, \ldots, x_n, T_1 \cdot y_1, \ldots, T_m \cdot y_m) = R_j(x_1, \ldots, x_n, y_1, \ldots, y_m).
    $$
    \item All the properties mentioned in the claim are true for our derivation $\{R_1', \ldots, R_p'\}$.
\end{itemize}

\mbox{}\\ \noindent{\textit{Base case: }} If $i = 1$ then $R_i \in \Gamma'$. If $R_i \in \Gamma$ then we can take $R_1' = R_1$. Otherwise, if $R_i = y_j - Q_j(\vec x)$ then we can take $R_1' = y_j' - Q_j'(\vec x, y_1', \ldots, y_{j - 1}')$ and $R_2' = \frac{y_j' - Q_j'(\vec x, y_1', \ldots, y_{j - 1}')}{T_j}$. Then it's obvious that 
$$
R_2'(\vec x, T_1 \cdot y_1, \ldots, T_m \cdot y_m) = R_1(\vec x, y_1, \ldots, y_m).
$$

\mbox{}\\ \noindent{\textit{Induction step: }} Suppose we have already constructed the $\mathsf{PC}_{\mathbb{Q}}^{\surd}$ refutation $\{R_1', R_2', \ldots, R_{p}'\}$ for which the induction statement is true. Now we have five cases depending on the way the $R_{i + 1}$ is derived.

\noindent{\textbf{Case 1}:} If $R_{i + 1} \in \Gamma'$ then this case is equivalent to the base case of induction. 

\noindent{\textbf{Case 2}:} If $R_{i + 1} = \alpha R_{j} + \beta R_l$ then $R_{p + 1}' = \alpha R_{j'}' + \beta R_{l'}'$ where $R_{j'}'(x_1, \ldots, x_n, T_1 \cdot y_1, \ldots, T_m \cdot y_m) = R_j(x_1, \ldots, x_n, y_1, \ldots, y_m)$ and $R_{l'}'(x_1, \ldots, x_n, T_1 \cdot y_1, \ldots, T_m \cdot y_m) = R_l(x_1, \ldots, x_n, y_1, \ldots, y_m)$. 

\noindent{\textbf{Case 3}:} If $R_{i + 1} = x_l \cdot R_j$ then $R_{p + 1}' = x_l \cdot R_{j'}'$ where 
$R_{j'}'(x_1, \ldots, x_n, T_1 \cdot y_1, \ldots, T_m \cdot y_m) = R_j(x_1, \ldots, x_n, y_1, \ldots, y_m)$. 

\noindent{\textbf{Case 4}:} If $R_{i + 1}^2 = R_j$ then we take $R_{p + 1}'(x_1, \ldots, x_n, y_1', \ldots, y_m') = R_{i + 1}(x_1, \ldots, x_n, \frac{y_1'}{T_1}, \ldots, \frac{y_m'}{T_m})$. By the induction statement we know that 
$$
R_{j}(x_1, \ldots, x_n, y_1, \ldots, y_m) = R_{j'}'(x_1, \ldots, x_n, T_1 \cdot y_1', \ldots, T_m \cdot y_m')
$$
for some $R_{j'}'$. Thus we know that 
$$
R_{j}(x_1, \ldots, x_n, \frac{y_1'}{T_1}, \ldots, \frac{y_m'}{T_m}) = R_{j'}'(x_1, \ldots, x_n, y_1', \ldots,y_m').
$$
So we know that 
\begin{multline*}
    R_{p + 1}'(x_1, \ldots, x_n, y_1', \ldots, y_m')^2 = R_{i + 1}(x_1, \ldots, x_n, \frac{y_1'}{T_1}, \ldots, \frac{y_m'}{T_m})^2  = \\ = R_{j}(x_1, \ldots, x_n, \frac{y_1'}{T_1}, \ldots, \frac{y_m'}{T_m}) =  R_{j'}'(x_1, \ldots, x_n, y_1', \ldots,y_m')
\end{multline*}
and $R_{p + 1}'$ is derived from $R_{j'}'$.

\noindent{\textbf{Case 5}:} If $R_{i + 1} = y_l \cdot R_j$ then we take $R_{p + 1}' = y_l' \cdot R_{j'}'$ and $R_{p + 2}' = \frac{R_{p + 1}'}{T_l}$ where $R_{j'}'(x_1, \ldots, x_n, T_1 \cdot y_1, \ldots, T_m \cdot y_m) = R_j(x_1, \ldots, x_n, y_1, \ldots, y_m)$. 

It's easy to see that in all these cases the induction statement stays true. 
\end{proofclaim}

\end{document}